\def\Statusstring{Technical Report\\21 December 2011}
\documentclass[12pt,letter]{article}

\usepackage{amsmath, amsthm, amssymb}
\usepackage{amssymb,epsfig}
\usepackage{pstricks} 

\newcommand{\vs}{1.5pt}  
\newcommand{\et}{1pt}  

\makeatletter

\gdef\@punct{.\ \ }  
\def\@sect#1#2#3#4#5#6[#7]#8{%
  \ifnum #2>\c@secnumdepth
     \def\@svsec{}
  \else
     \refstepcounter{#1}\edef\@svsec{%
     \ifnum #2>0{{\csname the#1\endcsname}}.\fi%
    \hskip .5em}
  \fi
  \@tempskipa #5\relax
  \ifdim \@tempskipa>\z@
     \begingroup #6\relax
       \@hangfrom{\hskip #3\relax\@svsec}{\interlinepenalty \@M #8\par}
     \endgroup
     \csname #1mark\endcsname{#7}
     \addcontentsline{toc}{#1}{\ifnum #2>\c@secnumdepth\else
          \protect\numberline{\csname the#1\endcsname}\fi#7}
  \else
     \def\@svsechd{#6\hskip #3\@svsec #8\@punct\csname
#1mark\endcsname{#7}
     \addcontentsline{toc}{#1}{\ifnum #2>\c@secnumdepth \else
          \protect\numberline{\csname the#1\endcsname}\fi#7}}
  \fi
  \@xsect{#5}}

\def\@ssect#1#2#3#4#5{\@tempskipa #3\relax
  \ifdim \@tempskipa>\z@
     \begingroup #4\@hangfrom{\hskip #1}{\interlinepenalty \@M
#5\par}\endgroup
  \else \def\@svsechd{#4\hskip #1\relax #5\@punct}\fi
  \@xsect{#3}}

\def\qed{\hskip 3pt \hbox{\vrule width4pt depth2pt height6pt}}

\newtheorem{Lemma}{Lemma}

\newtheorem{Theorem}[Lemma]{Theorem}

\newtheorem{Corollary}[Lemma]{Corollary}
\newtheorem{Definition}[Lemma]{Definition}

\newcommand{\diam}{\mathop{\mathrm{diam}}\nolimits}
\newcommand{\dist}{\mathop{\mathrm{dist}}\nolimits}
\newcommand{\inv}{\mathop{\mathrm{inv}}\nolimits}
\newcommand{\Fix}{\mathop{\mathrm{Fix}}\nolimits}

\begin{document}

\title{Diameter of Cayley graphs of permutation groups generated by transposition trees}
\author{Ashwin Ganesan%
  \thanks{Department of Mathematics, Amrita School of Engineering, Amrita Vishwa Vidyapeetham, Amritanagar, Coimbatore - 641~112, India.
  Email: \texttt{ashwin.ganesan@gmail.com}. }}
\date{}

\maketitle

\vspace{-6.5cm}
\begin{flushright}
  \texttt{\Statusstring}\\[1cm]
\end{flushright}
\vspace{+4.3cm}

\begin{abstract}
Let $\Gamma$ be a Cayley graph of the permutation group generated by a transposition tree $T$ on $n$ vertices. In an oft-cited paper \cite{Akers:Krishnamurthy:1989} (see also \cite{Hahn:Sabidussi:1997}), it is shown that the diameter of the Cayley graph $\Gamma$ is bounded as
$$\diam(\Gamma)~\le~\max_{\pi \in S_n}\left\{c(\pi)-n+\sum_{i=1}^n \dist_T(i,\pi(i))\right\},$$
where the maximization is over all permutations $\pi$, $c(\pi)$ denotes the number of cycles in $\pi$, and $\dist_T$ is the distance function in $T$. In this work, we first assess the performance (the sharpness and strictness) of this upper bound. We show that the upper bound is sharp for all trees of maximum diameter and also for all trees of minimum diameter, and we exhibit some families of trees
for which the bound is strict.  We then show that for every $n$, there exists a tree on $n$ vertices, such that the difference between the upper bound and the true diameter value is at least $n-4$.

Observe that evaluating this upper bound requires on the order of $n!$ (times a polynomial) computations.  We provide an algorithm that obtains an estimate of the diameter, but which requires only on the order of (polynomial in) $n$ computations; furthermore, the value obtained by our algorithm is less than or equal to the previously known diameter upper bound.  This result is possible because our algorithm works directly with the transposition tree on $n$ vertices and does not require examining any of the permutations (only the proof requires examining the permutations).  For all families of trees examined so far, the value $\beta$ computed by our algorithm happens to also be an upper bound on the diameter, i.e.
 $$\diam(\Gamma)~\le~\beta~\le~\max_{\pi \in S_n}\left\{c(\pi)-n+\sum_{i=1}^n \dist_T(i,\pi(i))\right\}.$$
\end{abstract}

\bigskip
\noindent\textbf{Index terms} --- Cayley graphs; transposition trees; algorithms; diameter; permutations; permutation groups; interconnection networks; parallel and distributed computing.
\bigskip


\newpage
\tableofcontents

\section{Introduction}

A problem of practical and theoretical interest is to determine or estimate the diameter of various families of Cayley networks. In the field of parallel and distributed computing, Cayley graphs generated by transposition trees were studied in the oft-cited paper by Akers and Krishnamurthy \cite{Akers:Krishnamurthy:1989}, where it was shown that the diameter of some families of Cayley graphs is sublogarithmic in the number of vertices.  This is one of the main reasons such Cayley graphs were considered a superior alternative to hypercubes for consideration as the topology of interconnection networks  \cite{Lakshmivarahan:etal:1993}.  Since then, much work has been done in understanding the use of Cayley graphs and algebraic methods in networks \cite{Heydemann:1997}.

Cayley graphs were first introduced in 1878 by Arthur Cayley to study groups \cite{Cayley:1878}.  Let $G$ be a group generated by a set of elements $S$.  The Cayley graph (or Cayley diagram) of $G$ with respect to the set of generators $S$ is a directed graph with vertex set $G$ and with an arc from vertex $g$ to vertex $gs$ iff $g \in G$ and $s \in S$ (cf. \cite[Chapter VIII.1]{Bollobas:1998}, \cite{Biggs:1993}).  Such graphs possess a certain degree of symmetry. In a symmetric network, the topology of the network looks the same from every node.  It is now known that many families of symmetric networks possess additional desirable properties such as optimal fault-tolerance \cite{Alspach:1992}, \cite{Akers:Krishnamurthy:1987}, algorithmic efficiency \cite{Annexstein:etal:1990}, optimal gossiping protocols \cite{Bermond:etal:1996} \cite{Zhou:2009}, and optimal routing algorithms \cite{Chen:etal:2006}, among others, and so have been widely studied in the fields of interconnection networks and distributed computing and networking.  New topologies continue to be proposed and assessed \cite{Parhami:2005}.  

The diameter of a network represents the maximum communication delay between two nodes in the network. The design and performance of bounds or algorithms that determine or estimate the diameter of various families of Cayley graphs of permutation groups is of much theoretical and practical interest. This diameter problem is difficult even for the simple case when the symmetric group is generated by cyclically adjacent transpositions (i.e. a set of transpositions whose transposition graph is a Hamilton cycle) \cite{Jerrum:1985}. The problem of determining the diameter of a Cayley network is the same as that of determining the diameter of the corresponding group for a given set of generators; the latter quantity is defined to be the minimum length of an expression for a group element in terms of the generators, maximized over all group elements. 

Let $S$ denote a set of transpositions of $\{1,2,\ldots,n\}$.  We can describe $S$ by its transposition graph $T(S)$, which is a simple, undirected graph whose vertex set is $\{1,2,\ldots,n\}$ and with vertices $i$ and $j$ being adjacent whenever $(i,j) \in S$.  Construct the Cayley graph $\Gamma$ whose vertex set is the permutation group generated by $S$ and with two vertices $g$ and $h$ being adjacent in $\Gamma$ if and only if there exists an $s \in S$ such that $gs=h$.  A natural problem is to understand how properties of the Cayley graph $\Gamma$ depend on those of the underlying transposition graph $T(S)$. For example, \cite[p. 53]{Godsil:Royle:2001} and \cite{Feng:2006} express the automorphism group of the Cayley graph in terms of that of the transposition tree.  In \cite{Akers:Krishnamurthy:1989}, \cite[p.188]{Heydemann:1997}, the diameter of the Cayley graph is expressed in terms of distances in the underlying transposition tree.

It is well known that a given set of transpositions on $\{1,2,\ldots,n\}$ generates the entire symmetric group on $n$ letters iff the transposition graph contains a spanning tree \cite{Berge:1971},\cite{Godsil:Royle:2001}.  A transposition graph which is a tree is called a transposition tree.   Throughout this work, we let $T=T(S)$ denote a given transposition tree corresponding to the set of transpositions $S$; we often use the same symbol $T$ to represent both the graph of the tree as well as a set of transpositions, and the notation $(i,j)$ is used to represents both an edge of $T$ as well as the corresponding transposition.  Since each element of $S$ is its own inverse, we have that  the Cayley graph $\Gamma$ generated by $T$ is a simple, undirected graph.  Let $\dist_G(u,v)$ denote the distance between vertices $u$ and $v$ in an undirected graph $G$, and let $\diam(G)$ denote the diameter of $G$.  Note that $\dist_{\Gamma}(\pi,\sigma) = \dist_{\Gamma}(I,\pi^{-1}\sigma)$, where $I$ denotes the identity permutation.  Thus, the diameter of $\Gamma$ is the maximum of $\dist_{\Gamma}(I,\pi)$ over $\pi \in S_n$.

The oft-cited paper \cite{Akers:Krishnamurthy:1989} gives an upper bound on the diameter of Cayley graphs generated by transposition trees.  When a bound or algorithm is proposed in the literature, it is often of interest to determine how far away the bound can be from the true value in the worst case, and to obtain more efficient algorithms for estimating the parameters.  The purpose of this work is to assess the performance of the previously known upper bound on the diameter and to propose a new algorithm to estimate the diameter of Cayley graphs for any given transposition tree.

\subsection{Notational preliminaries and terminology}

Let $S_n$ denote the symmetric group on $[n]:=\{1,2,\ldots,n\}$.  We represent a permutation $\pi \in S_n$ as an arrangement of $[n]$, as in $[\pi(1),\pi(2),\ldots,\pi(n)]$ or in cycle notation.  $c(\pi)$ denotes the number of cycles in $\pi$, including cycles of length 1.  Also, $\inv(\pi)$ denotes the number of inversions of $\pi$ (cf. \cite{Berge:1971}).  Thus, if $\pi = [3,5,1,4,2] = (1,3)(2,5) \in S_5$, then $c(\pi)=3$ and $\inv(\pi)=6$.  For $\pi,\tau \in S_n$, $\pi \tau$ is the permutation obtained by applying $\tau$ first and then $\pi$.  If $\pi \in S_n$ and $\tau = (i,j)$ is a transposition, then $c(\tau \pi) = c(\pi)+1$ if $i$ and $j$ are part of the same cycle of $\pi$, and $c(\tau \pi) = c(\pi)-1$ if $i$ and $j$ are in different cycles of $\pi$; and similarly for $c(\pi \tau)$.  $\Fix(\pi)$ denotes the set of fixed points of $\pi$, and $\overline{\Fix(\pi)}$ denotes the complement set $[n]-\Fix(\pi)$.  We assume throughout that $n \ge 5$, since the problem is easily solved by using brute force for all smaller trees.

Throughout this work, $\Gamma$ denotes the Cayley graph generated by a transposition tree $T$.  We now recall some previously known bounds and briefly outline the proof of these bounds.

\begin{Theorem} \label{thm:dist:ubound}  \cite{Akers:Krishnamurthy:1989}  Let $T$ be a tree and let $\pi \in S_n$.  Let $\Gamma$ be the Cayley graph generated by $T$.  Then
$$\dist_{\Gamma}(I,\pi) \le c(\pi)-n+\sum_{i=1}^n \dist_T(i,\pi(i)).$$
\end{Theorem}

By taking the maximum over both sides, it follows that
\begin{Corollary} \label{cor:diam:ubound} \cite[p.188]{Heydemann:1997}
$$\diam(\Gamma) \le \max_{\pi \in S_n} \left\{ c(\pi)-n+\sum_{i=1}^n \dist_T(i,\pi(i)) \right\}.$$
\end{Corollary}

In the sequel, we refer to the first upper bound as the \emph{distance upper bound} $f_T(\pi)$ and the second upper bound as the \emph{diameter upper bound} $f(T)$:

\begin{Definition}  For a transposition tree $T$ and $\pi \in S_n$, define
$$f(T): = \max_{\pi \in S_n} f_T(\pi) = \max_{\pi \in S_n} \left\{c(\pi)-n+\sum_{i=1}^n \dist_T(i,\pi(i))\right\}.$$
\end{Definition}

We now recall from \cite{Akers:Krishnamurthy:1989} the proofs of these results since we refer to the proof in the sequel.  Start with a given tree $T$ on vertices labeled by $[n]$ and an element $\pi \in S_n$, $\pi \ne I$, for which we wish to determine $\dist_{\Gamma}(I,\pi)$.  Initially, at each vertex $i$ of $T$, we place a marker $\pi(i)$.  Multiplying $\pi$ by the transposition $(i,j)$ amounts to \emph{switching} the markers at vertices $i$ and $j$.  We now have a new set of markers at each vertex of $T$ corresponding to the permutation $\pi(i,j)$, which is a vertex adjacent to $\pi$ in $\Gamma$. The problem of determining $\dist_{\Gamma}(I,\pi)$ is then equivalent to that of finding the minimum number of switches necessary to home each marker (i.e. to bring each marker $i$ to vertex $i$).   Given any $T$ with vertex set $[n]$ and markers for these vertices corresponding to $\pi \ne I$, it can be shown that $T$ always has an edge $ij$ such that the edge satisfies one of the following two conditions: Either (A) the marker at $i$ and the marker at $j$ will both reduce their distance to $\pi(i)$ and $\pi(j)$, respectively, if the switch $(i,j)$ is applied, or (B) the marker at one of $i$ or $j$ is already homed, and the other marker wishes to use the switch $(i,j)$.  We call an edge that satisfies one of these two conditions an {\em admissible edge} of type A or type B.  It can be shown that during each step that a transposition $\tau$ corresponding to an admissible edge is applied to $\pi$, we get a new vertex $\pi'$ which has a strictly smaller value of the left hand side above; i.e., $f_T(\pi') < f_T(\pi)$, and it can be verified that $f_T(I)=0$.  This proves the bounds above. This algorithm, which we call the AK algorithm, can be viewed as `sorting' a permutation using only the transpositions defined by $T$, and the Cayley graph is the state transition diagram of the current permutation of markers.

We point out that this same diameter upper bound inequality is also derived in Vaughan \cite{Vaughan:1991}; however, this paper was published in 1991, whereas Akers and Krishnamurthy \cite{Akers:Krishnamurthy:1989} was published in 1989 and widely picked up on in the parallel and distributed computing and networking community by then.  There are some subsequent papers, such as \cite{Vaughan:Portier:1995} and \cite{Smith:1999}, which cite only Vaughan \cite{Vaughan:1991} and not \cite{Akers:Krishnamurthy:1989}.

Note that the distance and diameter bounds above need not hold if $T$ has cycles (the proof mentioned above breaks down because if $T$ has cycles, there exists a $\pi \ne I$ such that $T$ has no admissible edges for this $\pi$).  Thus, when we study the sharpness (or lack thereof) of the upper bounds, we assume throughout that $T$ is a tree and $\Gamma$ is the Cayley graph generated by a tree.

The diameter of Cayley graphs generated by transposition trees is known for some particular families of graphs.  For example, if the transposition tree is a path graph on $n$ vertices, then the diameter of the corresponding Cayley graph is ${n \choose 2}$, and if the tranposition tree is a star $K_{1,n-1}$, then the diameter of the corresponding Cayley graph is $\lfloor 3(n-1)/2 \rfloor$ (cf. \cite{Akers:Krishnamurthy:1989}).  For the special case when $T$ is a star, another upper bound on the distance between vertices in the Cayley graph is known:
\begin{Lemma}~ \cite{Akers:Krishnamurthy:1989}Let $T$ be a star. Then
$$\dist_{\Gamma}(I,\pi) \le n+c(\pi)-2|\Fix(\pi)| - r(\pi),$$
where $r(\pi)$ equals 0 if $\pi(1)=1$ and $r(\pi)=2$ otherwise (here, the center vertex of $T$ is assumed to have the label 1).
\end{Lemma}

It is possible to obtain a heuristic derivation of the diameter upper bound formula, as follows.  It is straightforward to derive the distance upper bound for the special case when the transposition tree is a star $K_{1,n-1}$, and we get \cite{Akers:Krishnamurthy:1989}
$$\dist_{\Gamma}(I,\pi) \le n+c(\pi)-2|\Fix(\pi)| - r(\pi).$$
Observe that $|\Fix(\pi)| = n-|\overline{\Fix(\pi)|}$, which yields
$$\dist_{\Gamma}(I,\pi) \le c(\pi) - n + 2|\overline{\Fix(\pi)}| - r(\pi).$$
Note that when the tree is a star, $2 |\overline{\Fix(\pi)}|$ is almost (i.e. within 1 of) the sum of distances $\sum_{i=1}^n \dist_T(i,\pi(i))$.  This leads us to the question of whether the inequality
$$\dist_{\Gamma}(I,\pi) \le c(\pi)-n+\sum_{i=1}^n \dist_T(i,\pi(i))$$
also holds for all the remaining trees $T$, and this question has been answered affirmatively in the literature.

\subsection{Summary of our main contributions}

The problem of interest studied in this work is to obtain algorithms or bounds for the diameter of Cayley graphs of permutation groups.  When a bound or algorithm is proposed in the literature, it is of interest to determine the families of graphs for which the previously known bounds are exact, and to determine how far away these bounds can be from the true value in the worst case.  It is also of interest to have new or more efficient algorithms for estimating these parameters.  We now summarize our contributions on this problem.

%
%

Let $\Gamma$ denote the Cayley graph generated by a transposition tree $T$.  We show that the previously known distance upper bound
$$\dist_{\Gamma}(I,\pi) \le c(\pi)-n+\sum_{i=1}^n \dist_T(i,\pi(i)).$$
is exact or sharp (i.e. the inequality holds with equality) for all $\pi \in S_n$ if and only if $T$ is a star.

We also show that the previously known diameter upper bound
$$\diam(\Gamma) \le \max_{\pi \in S_n} \left\{ c(\pi)-n+\sum_{i=1}^n \dist_T(i,\pi(i)) \right\}.$$
is exact if $T$ is a star or a path.  Note that this also implies that even though the distance upper bound is {\em not} exact for any paths, if we take the maximum over both sides, we get that the diameter upper bound {\em is} exact for all paths (i.e. the strict inequality becomes an equality when we maximize over all $\pi \in S_n$).

It was shown in \cite{Akers:Krishnamurthy:1989} that: when $T$ is a star,
$$\dist_{\Gamma}(I,\pi) \le n+c(\pi)-2|\Fix(\pi)| - r(\pi).$$
We show here that this inequality holds with equality.

We then examine some properties of the AK algorithm.  We show that if the tree $T$ is a star or a path, then any factorization obtained by the AK algorithm to express a given permutation $\pi \in S_n$ as a word in the edges of $T$ is of minimum length.  However, there exist other transposition trees for which the AK algorithm is not optimal.

It is of interest to know how far away a bound can be from the true value in the worst case.  We show that for every $n$, there exists a transposition tree on $n$ vertices such that the difference between the diameter upper bound and the true diameter value of the Cayley graph is at least $n-4$.  This result gives a lower bound on the difference, and we leave it as an open problem to determine an upper bound for this difference.

Observe that evaluating the diameter upper bound requires on the order of $n!$ (times a polynomial) computations.  We propose another algorithm for estimating the diameter of the Cayley graph of the permutation group.  Our algorithm obtains an estimate, say $\beta$, efficiently, using only on the order of $n$ (times a polynomial) computations.  This is possible because our algorithm works directly with the transposition tree on $n$ vertices and does not require examining the different permutations (it is only the proof of this algorithm that requires examining the permutations). Furthermore, the estimate obtained by our algorithm is shown to be less than or equal to the previously known diameter upper bound.  Also, for all families of trees investigated so far, the estimate is an upper bound on the diameter, i.e.
 $$\diam(\Gamma)~ \le~ \beta ~ \le ~ \max_{\pi \in S_n} \left\{ c(\pi)-n+\sum_{i=1}^n \dist_T(i,\pi(i)) \right\},$$
 where we show that the second inequality holds for all trees, and the first inequality holds for many families of trees (in fact for all families of trees investigated so far).

Some further interesting extensions and open problems on this algorithm and related bounds are discussed in Section~\ref{sec:algorithm}.

\section{Sharpness of the distance and diameter upper bounds}\label{sec:mainresults}
In our proofs, it will be convenient to define the sum of distances in the tree for a given permutation $$S_T(\pi) := \sum_{i=1}^n \dist_T(i,\pi(i)).$$  Thus, $f_T(\pi) = c(\pi)-n+S_T(\pi)$.
While the bounds and results here are independent of the labeling of the vertices of $T$, it will be convenient to assume that the center vertex of a star has label 1, and that the vertices of a path are labeled consecutively from 1 to $n$.
Also, note that the diameter and distance bounds are invariant to a translation of the labels on the set of integers, i.e. we can replace the labels $\{1,2,\ldots,n\}$ by say $\{2,3,\ldots,n+1\}$.

\bigskip
\begin{Theorem} \label{Thm:iffStar} Let $\Gamma$ be the Cayley graph generated by transposition tree $T$. Then, in the distance upper bound inequality
$$ \dist_{\Gamma}(I,\pi) \le c(\pi)-n+\sum_{i=1}^n \dist_T(i,\pi(i)),$$
we have equality for all $\pi \in S_n$ if and only if $T$ is the star $K_{1,n-1}$.
\end{Theorem}

\begin{proof}
Suppose $T$ is the star $K_{1,n-1}$.  It is already known that $\dist_{\Gamma}(I,\pi) \le f_T(\pi)$ for all $\pi \in S_n$.  We now prove the reverse inequality.  We want to show that $f_T(\pi)$ is the minimum number of transpositions of the form $(1,i), 2 \le i \le n$ required to sort $\pi$.  Each vertex $i$ of $T$ is initially assigned the marker $\pi(i)$.  Before the markers along edge $(1,i)$ are interchanged,  there are four possibilities for the values of the marker $\pi(1)$ at vertex 1 and marker $\pi(i)$ at vertex $i$:

(a)  $\pi(1)=1$ and $\pi(i)=i$:  In this case, applying transposition $(1,i)$ creates a new permutation $\pi'$ for which $c(\pi)$ has reduced by 1 (i.e. $c(\pi')=c(\pi)-1$) and $S_T(\pi)$ has increased by 2, thereby increasing $f_T(\pi)$ by 1.

(b) $\pi(1)=1$ and $\pi(i) \ne i$:  Applying transposition $(1,i)$ reduces $c(\pi)$ and doesn't affect $S_T(\pi)$.  Hence, $f_T(\pi)$ is reduced by 1.

(c) $\pi(1) \ne 1$ and $\pi(i)=i$:  Applying $(1,i)$ reduces $c(\pi)$ by 1 and increases $S_T(\pi)$ by 2, hence increases $f_T(\pi)$ by 1.

(d)  $\pi(1) \ne 1$ and $\pi(i) \ne i$:  There are four subcases here:

(d.1) The first subcase is when $\pi(1)=i$ and $\pi(i)=1$.  In this case, applying $(1,i)$ increases $c(\pi)$ by 1, and reduces $S_T(\pi)$ by 2, thereby reducing $f_T(\pi)$ by 1.

(d.2)  Suppose $\pi(1)=j$ (where $j \ne 1,i$) and $\pi(i)=1$. Applying $(1,i)$ increases $c(\pi)$ by 1 and doesn't affect $S_T(\pi)$.  So $f_T(\pi)$ increases by 1.

(d.3) Suppose $\pi(1)=i$ and $\pi(i)=j \ne 1$.  Then applying $(1,i)$ increases $c(\pi)$ by 1 and reduces $S_T(\pi)$ by 2, and hence reduces $f_T(\pi)$ by 1.

(d.4)  Suppose $\pi(1)=k$ and $\pi(i)=j$, where $j,k \ne 1,i$.  Then applying $(1,i)$ changes $c(\pi)$ by 1 and doesn't change $S_T(\pi)$, so that $f_T(\pi)$ changes by 1.

In all cases above, switching the markers on an edge $(1,i)$ of $T$ reduces $f_T(\pi)$ by at most 1.  Hence, the minimum number of transpositions required to sort $\pi$, or equivalently, the value of $\dist_{\Gamma}(I,\pi)$, is at least $f_T(\pi)$.  Hence, $\dist_{\Gamma}(I,\pi) \ge f_T(\pi)$ for all $\pi \in S_n$.   This proves the reverse inequality.

Observe that $S_T(\pi)=2|\overline{\Fix(\pi)}|$ when $\pi(1)=1$, and $S_T(\pi)=2|\overline{\Fix(\pi)}|-2$ otherwise. Thus, it is seen that when $T$ is the star graph, $f_T(\pi)$ evaluates to $c(\pi) - n + 2 |{\overline{\Fix(\pi)}}| - r(\pi).$  Hence, $\dist_{\Gamma}(I,\pi) = c(\pi) - n + 2 |{\overline{\Fix(\pi)}}| - r(\pi)$.

Now suppose $T$ is not a star.  Then, $\diam(T) \ge 3$. So $T$ contains 4 ordered vertices $i,j,k$ and $\ell$ that comprise a path of length 3.  Let $\pi$ be the permutation $(i,k)(j,l)$.  Then, $c(\pi) = n-2$, $S_T(\pi)=8$, and hence, $f_T(\pi)=6$, but $\dist_{\Gamma}(I,\pi) \le 4$, as can be easily verified by applying transpositions $(j,k),(i,j),(k,l)$ and $(j,k)$.
\end{proof}

\bigskip
\begin{Corollary} \label{cor:star} Let $T$ be the star $K_{1,n-1}$ on $n$ vertices.  Then the previously known upper bound inequalities
$$\dist_{\Gamma}(I,\pi) \le n+c(\pi)-2|\Fix(\pi)| - r(\pi),$$
and
$$\diam(\Gamma) \le \max_{\pi \in S_n} \left\{ c(\pi)-n+\sum_{i=1}^n \dist_T(i,\pi(i)) \right\}$$
hold with equality.
\end{Corollary}

Theorem~\ref{Thm:iffStar} implies that if $T$ is the path graph (which is not a star for $n \ge 4$), then there exists a $\pi \in S_n$ for which the distance upper bound is strict:
\begin{Corollary} \label{cor:star2} Let $T$ be the path graph on $n$ vertices.  Then there exists a $\pi \in S_n$ for which
$$\dist_{\Gamma}(I,\pi) < c(\pi)-n+\sum_{i=1}^n \dist_T(i,\pi(i)).$$
\end{Corollary}

  Despite such a result, when taking the maximum over both sides, we obtain equality:

\bigskip
\begin{Theorem}  \label{Thm:pathfTmax:equals:diam} Let $\Gamma$ be the Cayley graph generated by a transposition tree $T$. Then the diameter upper bound inequality
$$\diam(\Gamma) \le \max_{\pi \in S_n} \left\{ c(\pi)-n+\sum_{i=1}^n \dist_T(i,\pi(i)) \right\}$$
holds with equality if $T$ is a path.
\end{Theorem}

\begin{proof}
Let $T$ be the path graph on $n$ vertices.  It is known that $\diam(\Gamma) = {n \choose 2}$ \cite{Berge:1971}.  Hence, it suffices to prove
$$\max_{\pi \in S_n} f_T(\pi) = {n \choose 2}.$$
Let $\sigma = [n, n-2, \ldots, 2,1]$.  It can be verified that $f_T(\sigma)$ evaluates to ${n \choose 2}$.  Thus, it remains to prove the bound
$$f_T(\pi) \le {n \choose 2},~~~\forall~\pi \in S_n.$$
We prove this by induction on $n$.  The assertion can be easily verified for small values of $n$.  So fix $n$, and now assume the assertion holds for smaller values of $n$.   Write $\pi$ as $\pi_1 \pi_2 \ldots \pi_s$.  Thus $\pi$ is a product of $s$ disjoint cycles, and suppose $\pi_1$ is the cycle that contains $n$. We consider three cases, depending on whether $\pi$ fixes $n$, whether $\pi$ maps $n$ to 1, or  whether $\pi$ maps $n$ to some $j \ne 1$:

(a) Suppose $\pi_1 = (n)$.  Define $\pi' := \pi_2 \ldots \pi_s \in S_{n-1}$. Let $T'$ be the tree on vertex set $[n-1]$.  We have that  $f_T(\pi) = c(\pi) - n + S_T(\pi) = c(\pi')+1 -n +S_{T'}(\pi') = c(\pi') - (n-1) + S_{T'}(\pi')$, which is at most ${{n-1} \choose 2}$ by the inductive hypothesis.

(b) Suppose $\pi$ maps $n$ to 1.  There are a few subcases, depending on the length of $\pi_1$:

(b.1) Suppose $\pi_1 = (n,1)$, a transposition.  Define $\pi' := \pi(n,1) = (1)(n)\pi_2 \ldots \pi_s$.  Then, $c(\pi') = c(\pi)+1$ and $S_T(\pi') = S_T(\pi)-2(n-1)$.  Hence, $f_T(\pi) = c(\pi)-n+S_T(\pi) = c(\pi')-1-n+S_T(\pi')+2(n-1)$.  Let $T''$ denote the tree on vertex set $\{2,3,\ldots,n-1\}$, and let $\pi'' = \pi_2 \ldots \pi_s$ be a  permutation of the vertices of $T''$.  Then $f_T(\pi) = 2+c(\pi'')-1-n+S_T(\pi')+2(n-1)=c(\pi'')-(n-2)+S_{T''}(\pi'')+2(n-1)-1$.  Relabeling the vertices of $T''$ and the elements of $\pi''$ from $\{2,\ldots,n-1\}$ to $[n-2]$ does not change $c(\pi'')-(n-2)+S_{T''}(\pi'')$, to which we can apply the inductive hypothesis.  The bound then follows.

(b.2) Suppose $\pi_1 = (n,1,j)$, where $2 \le j \le n-1$.  Define $\pi'=(n,1)(j)\pi_2 \ldots \pi_s$.  It can be verified that $S_T(\pi)=S_T(\pi')$, and $c(\pi')=c(\pi)+1$.  Hence, $f_T(\pi) = f_T(\pi')-1 \le {n \choose 2}-1$ by the earlier subcase (b.1).  Note that $S_T(\pi) = S_T(\pi^{-1})$ and $c(\pi)=c(\pi^{-1})$, so that the bound evaluates to the same value when $\pi_1=(n,j,1)$ and when $\pi_1=(1,n,j)$.

(b.3) Suppose $\pi_1=(n,1,j_1,\ldots,j_\ell)$ contains at least 4 elements. Define $\pi' = (n,1)(j_1,\ldots,j_\ell)\pi_2 \ldots \pi_s$, where the first cycle of $\pi$ has now been broken down into two disjoint cycles to obtain $\pi'$.  Define $x:=|j_1-j_2|+\ldots+|j_{\ell-1}-j_\ell|$. Then, $S_T(\pi')=2(n-1)+x+|j_\ell-j_1|+d$ for some $d$, where $d$ is the sum of distances obtained from the remaining cycles $\pi_2,\ldots,\pi_s$.   Also, $S_T(\pi)=n-1+|1-j_1|+x+|j_\ell-n|+d$.  Also, $c(\pi)=c(\pi')-1$.  Using the equations obtained here and substituting, we get that $f_T(\pi) = c(\pi)-n+S_T(\pi) = c(\pi')-n+S_T(\pi')-n+|1-j_1|+|j_\ell-n|-|j_\ell-j_1|$.  Using the bound $c(\pi')-n+S_T(\pi') \le {n \choose 2}$ of subcase (b.1), and using the fact that $|1-j_1|+|j_\ell-n|-|j_\ell-j_1| \le n-1$, we get the desired bound $f_T(\pi) \le {n \choose 2}$.

(c) We consider two subcases.  In the first subcase, 1 and $n$ are in the same cycle of $\pi$, and in the second subcase 1 and $n$ are in different cycles of $\pi$.

(c.1) Let $\pi = (n,j_1,\ldots,j_\ell,1,k_1,\ldots,k_t)\pi_2 \ldots \pi_s$.

Define $\pi' =  (n,j_1,\ldots,j_\ell,1)$ $(k_1,\ldots,k_t)\pi_2 \ldots \pi_s$.  We show that $f_T(\pi) \le f_T(\pi')$. This latter quantity is bounded from above by ${n \choose 2}$ due to the earlier subcases.

(c.1.1) The subcase $\ell=0$ has been addressed in subcases (b.2) and (b.3).  Since there is a vertex automorphism of the path graph $T$ that maps 1 to $n$ and $n$ to 1, the subcase $t=0$ has also been addressed by the subcase $\ell=0$.

(c.1.2)  Now suppose $\ell=1$ and $t=1$.  The sum of distances of elements in the cycle $(n,j_1,1,k_1)$ and $(n,j_1,1)(k_1)$ are equal, and $c(\pi) < c(\pi')$.  Hence, $f_T(\pi) < f_T(\pi')$.

(c.1.3) Now suppose $t=1$ and $l \ge 2$; note that by symmetry, this subcase also addresses the subcase $\ell=1$ and $t \ge 2$.  A calculation of the sum of distances $S_T(\pi)$ and $S_T(\pi')$ yields, again, that $S_T(\pi)=S_T(\pi')$.  Since $c(\pi) < c(\pi')$, $f_T(\pi) < f_T(\pi')$

(c.1.4) Finally, suppose $t \ge 2$ and $\ell \ge 2$. Recall that $\pi$ contains the cycle $(n,j_1,\ldots,j_\ell,1,k_1,\ldots,k_t)$, and $\pi'$ contains the two cycles  $(n,j_1,\ldots,j_\ell,1)$ and $(k_1,\ldots,k_t)$.  A summation of distances due to elements in these cycles yields that $f_T(\pi) \le f_T(\pi')$ if and only if $k_1-k_t \le |k_1-k_t|+1$, which is clearly true.

(c.2) Let $\pi = (n,j_1,\ldots,j_\ell)(1,k_1,\ldots,k_t)\pi_3 \ldots \pi_s$.

(c.2.1) Suppose $\ell=1$, i.e. $\pi_1=(n,j_1)$ is a cycle of $\pi$.  Define a new permutation $\pi' = (n,1) \pi_2' \ldots \pi_s'$ that has the same type as $\pi$ but with the labels of 1 and $j_1$ interchanged, i.e. $\pi' = (1,j_1)~\pi~(1,j_1)$. Then $c(\pi')=c(\pi)$.  Note that $S_T(\pi)$ contains terms $|n-j_1|$ and $|j_1-n|$, corresponding to the cycle $\pi_1$.  When going from $\pi$ to $\pi'$, the sum of two terms of $S_T$ is increased by an amount equal to $2|j_1-1|$  because the cycle $(n,j_1)$ is replaced by the cycle $(n,1)$.   When going from $\pi$ to $\pi'$, the cycle containing the element 1 is now replaced by a cycle containing the element $j_1$, and this could contribute to a decrease in $S_T$ by at most $2|j_1-1|$.  Hence, $f_T(\pi) \le f_T(\pi')$.  The bound then follows from applying the earlier subcase (b.1) to $\pi'$.  This resolves the case $\ell=1$, and by symmetry, also the case $t=1$.

(c.2.2)  So now assume $\ell \ge 2$ and $t \ge 2$.

Let $\pi = (n,j_1,\ldots,j_\ell)(1,k_1,\ldots,k_t)\pi_3 \ldots \pi_s$, and

let $\pi' = (n,1)(j_1,\ldots,j_\ell,k_1,\ldots,k_t)$ $\pi_3 \ldots \pi_s$.
A computation of the sum of distances in $S_T(\pi)$ and $S_T(\pi')$ yields that $f_T(\pi) \le f_T(\pi')$ if and only if $k_1-j_\ell+k_t-j_1 \le |k_1-j_\ell|+|k_t-j_1|+1$, which is clearly true.
\end{proof}
\section{On the AK algorithm}

We describe some properties of the AK algorithm here.

\begin{Theorem} \label{Thm:AK:optimal}
If $T$ is a star or a path, then the AK algorithm sorts any permutation using the minimum number of transpositions.
\end{Theorem}

\begin{proof}
Let $T$ be the path graph, with the vertices labeled consecutively from 1 to $n$. Let $\pi \in S_n$ be a given permutation.  Then, the AK algorithm chooses, during each step, an admissible edge $(i,i+1)$ of type A or type B.  If the edge is of type A, then the marker $\pi(i)$ at vertex $i$ reduces its distance to vertex $\pi(i)$ if the transposition $(i,i+1)$ is applied, and similarly for the marker $\pi(i+1)$ at vertex $i+1$.  Hence, by the chosen labeling of the vertices, $\pi(i) > \pi(i+1)$.  Thus, applying $(i,i+1)$ reduces the number of inversions of the given permutation by 1.  Similarly, if the edge is of type B, applying $(i,i+1)$ reduces again the number of inversions of the given permutation by 1.  Thus, in either case, after $(i,i+1)$ is applied to $\pi$, we get a new permutation which has exactly one fewer inversions than $\pi$.  Thus, the AK algorithm uses exactly $\inv(\pi)$ transpositions to home all the markers, and it is a well-known result that this is the minimum number $\dist_{\Gamma}(I,\pi)$ of transpositions possible.

Let $T$ be the star. The different cases in the proof of Theorem~\ref{Thm:iffStar} were (a),(b),(c) and (d.1) to (d.4).  Each time a transposition is applied by the AK algorithm, it picks an admissible edge of type A or type B.  If the edge is of type A, then we are in case (d.1) or (d.3), in which case $f_T(\pi)$ surely reduces by 1.  If the edge is of type B, then we are in case (b), in which case $f_T(\pi)$ again surely reduces by 1.  Thus, the AK algorithm sorts $\pi$ using exactly $f_T(\pi)$ transpositions of $T$ and this is the minimum possible number of transpositions by Theorem~\ref{Thm:iffStar}.
\end{proof}

\bigskip
\begin{Theorem} \label{Thm:AK:suboptimal}
There exist transposition trees for which the diameter upper bound is strict.  There exist transposition trees for which the AK algorithm uses more than the minimum number of transpositions required.
\end{Theorem}

\begin{proof}
Let $T$ be the transposition tree on 5 vertices consisting of the 4 transpositions $(1,2),(2,3),(1,4)$ and $(1,5)$.  Let $\pi = (2,4)(3,5) \in S_5$. Then $f_T(\pi)=8$, whereas a quick computer simulation using GAP \cite{GAP4} confirms that the  diameter of the Cayley graph generated by $T$ is 7.  Hence, there exist transposition trees for which the diameter upper bound inequality is strict.

Next, suppose $T$ is the transposition tree on 7 vertices consisting of the 6 transpositions $(1,2),(2,3),(1,4),$ $(4,5),(1,6)$ and $(6,7)$.  Let $\pi=(2,4)(3,5)(5,7) \in S_7$.  Then the following 15 edges of $T$, when applied in the order given, are all admissible edges (of type A or type B), and can be used to sort $\pi$ on $T$: $(1,2),(1,4),(1,2),$ $(2,3),(1,2),(1,6),$ $(6,7),(1,2),(2,3),(4,5),$ $(1,4),(1,6),$ $(6,7),(1,4),(4,5)$.  However, it can be verified (with the help of a computer) that the diameter of the Cayley graph generated by $T$ is 14.  Hence, the AK algorithm can take more than the minimum required number of transpositions to sort a given permutation.
\end{proof}

\section{Strictness of the diameter upper bound}
Recall that the diameter of a Cayley graph $\Gamma$ generated by a transposition tree $T$ is bounded as
$$\diam(\Gamma) \le \max_{\pi \in S_n} \left\{ c(\pi)-n+\sum_{i=1}^n \dist_T(i,\pi(i)) \right\} =: f(T).$$
Let $f(T)$ denote the upper bound in the right hand side of the inequality.  When bounds for the performance of networks are proposed, it is of both theoretical and practical interest to investigate how far away this bound can be from the true value.  We now assess the performance of this bound and derive a strictness result.

Define the worst case performance of this upper bound by the quantity
$$\Delta_n := \max_{T \in \mathcal{T}_n} |f(T) - \diam(\Gamma)|,$$
where $\mathcal{T}_n$ denotes the set of all trees on $n$ vertices.

\begin{Theorem}
For every $n \ge 5$, there exists a tree on $n$ vertices such that the difference between the actual diameter of the Cayley graph and the diameter upper bound is at least $n-4$; in other words, $\Delta_n \ge n-4.$
\end{Theorem}

\begin{proof}
Throughout this proof, we let $T$ denote the transposition tree defined by the edge set $\{(1,2),(2,3),\ldots,(n-3,n-2),(n-2,n-1),(n-2,n)\}$, which is shown in Figure~\ref{fig:tree:in:proof}.  For conciseness, we let $d(i,j)$ denote the distance in $T$ between vertices $i$ and $j$. Also, for leaf vertices $i,j$ of $T$, we let $T-\{i,j\}$ denote the tree on $n-2$ vertices obtained by removing vertices $i$ and $j$ of $T$.

Our proof is in two parts.  In the first part we establish that $f(T)$ is equal to ${n \choose 2}-2$.  In the second part we show that the diameter of the Cayley graph generated by $T$ is at most ${{n-1} \choose 2}+1$.  Together, this yields the desired result.
\begin{center}
\begin{figure}
\begin{pspicture}(-3,0)(5,2)
\qdisk(0,1){\vs}
\qdisk(1,1){\vs}
\qdisk(2,1){\vs}
\qdisk(3,1){\vs}
\qdisk(4,1){\vs}
\qdisk(5,2){\vs}
\qdisk(5,0){\vs}
\psline[linewidth=\et]{-}(0,1)(1,1)
\psline[linewidth=\et]{-}(1,1)(2,1)
\uput[270](2.5,1.2){$\ldots$}
\psline[linewidth=\et]{-}(3,1)(4,1)
\psline[linewidth=\et]{-}(4,1)(5,2)
\psline[linewidth=\et]{-}(4,1)(5,0)

\uput[270](0,1){$1$}
\uput[270](1,1){$2$}
\uput[270](2,1){$3$}
\uput[270](3,1){$n-3$}
\uput[0](4,1){$n-2$}
\uput[0](5,2){$n-1$}
\uput[0](5,0){$n$}
\end{pspicture}
\caption{A transposition tree $T$ on $n$ vertices.} \label{fig:tree:in:proof}
\end{figure}
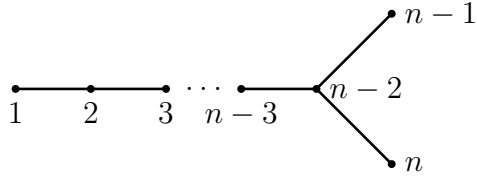
\end{center}

We now present the first part of the proof; we establish that $f(T)$, defined by
$$f(T) := \max_{\sigma \in S_n} \left\{ c(\sigma)-n+\sum_{i=1}^n \dist_T(i,\sigma(i)) \right\},$$
is equal to ${n \choose 2}-2$.  We prove this result by examining several sub-cases.  Define
$$f_T(\sigma) := c(\sigma)-n+S_T(\sigma),~~~S_T(\sigma):=\sum_{i=1}^n \dist_T(i,\sigma(i)).$$
We consider two cases, (1) and (2), depending on whether 1 and $n$ are in the same or different cycle of $\sigma$; each of these cases will further involve subcases.  In most of these subcases, we show that for a given $\sigma$, there is a $\sigma'$ such that $f_T(\sigma) \le f_T(\sigma')$ and $f_T(\sigma') \le {n \choose 2}-2$.

(1)  Assume 1 and $n$ are in the same cycle of $\sigma$.  So $\sigma=(1,k_1,\ldots,k_s,n,j_1,\ldots,j_\ell) \hat{\sigma}$.  The different subcases consider the different possible values for $s$ and $\ell$.

(1.1) Suppose $s=0,\ell=0$. So $\sigma=(1,n) \hat{\sigma} = (1,n) \sigma_2 \ldots \sigma_r$.
Then, $f_T(\sigma)=c(\sigma)-n+S_T(\sigma) = r-n+2(n-2)+S_{T-\{1,n\}}(\hat{\sigma})=2n-5+(r-2)+(n-2)+S_{T-\{1,n\}}(\hat{\sigma}) = 2n-5+c(\hat{\sigma})+(n-2)+S_{T-\{1,n\}}(\hat{\sigma}) = 2n-5+f_{T-\{1,n\}}(\hat{\sigma}) \le  2n-5+{{n-2} \choose 2} = {n \choose 2}-2$, where by Theorem~\ref{Thm:pathfTmax:equals:diam} the inequality holds with equality for some $\hat{\sigma}$. Thus, the maximum of $f_T(\sigma)$ over all permutations that contain $(1,n)$ as a cycle is equal to ${n \choose 2}-2$.  It remains to show that for all other kinds of permutations $\sigma$ in the symmetric group $S_n$, $f_T(\sigma) \le {n \choose 2}-2$.

(1.2) Suppose $s=1, \ell=0$.  So $\sigma=(1,i,n) \sigma_2 \ldots \sigma_r = (1,i,n) \hat{\sigma}$. We consider some subcases.

(1.2.1)  Suppose $i=n-1$.  Then, $f_T(\sigma) = r-n+(2n-2)+S_{T-\{1,n-1,n\}}(\hat{\sigma}) = 2n-4+f_{T-\{1,n-1,n\}}(\hat{\sigma}) \le 2n-4+{{n-3} \choose 2} \le {n \choose 2}-2$, where the inequality is by Theorem~\ref{Thm:pathfTmax:equals:diam}.

(1.2.2)  Suppose $2 \le i \le n-2$;  so $\sigma=(1,i,n) \hat{\sigma}$.  Let $\sigma' = (1,n)(i) \hat{\sigma}$.  It is easily verified that $f_T(\sigma) \le f_T(\sigma')$, and so the desired bound follows from applying subcase (1.1) to $f_T(\sigma')$.

 (1.3) Suppose $s=0, \ell=1$, so $\sigma=(1,n,i) \hat{\sigma}$. Since $f_T(\sigma) =  \ f_T(\sigma^{-1})$, this case also is settled by (1.2).

 (1.4)  Suppose $s=0, \ell \ge 2$, so $\sigma = (1,n,j_1,\ldots,j_\ell) \hat{\sigma}$.  Let $\sigma' = (1,n)(j_1,\ldots,j_\ell) \hat{\sigma}$. Observe that $f_T(\sigma) \le f_T(\sigma')$ iff $d(n,j_1)+d(j_\ell,1) \le d(n,1)+d(j_\ell,j_1)+1$.  We prove the latter inequality by considering 4 subcases:

(1.4.1) Suppose $j_1 < j_\ell \le n-2$.  Then, an inspection of the tree in Figure~\ref{fig:tree2:in:proof} shows that $d(n,j_1)+d(j_\ell,1) = d(n,1)+d(j_\ell,j_1)$, and so the inequality holds.
\begin{center}
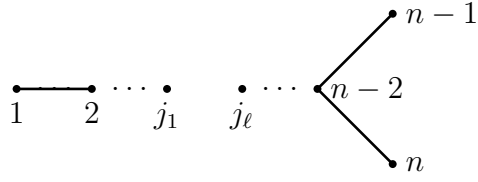
\begin{figure}
\begin{pspicture}(-3,0)(5,2)
\qdisk(0,1){\vs}
\qdisk(1,1){\vs}
\qdisk(2,1){\vs}
\qdisk(3,1){\vs}
\qdisk(4,1){\vs}
\qdisk(5,2){\vs}
\qdisk(5,0){\vs}
\psline[linewidth=\et]{-}(0,1)(1,1)
\uput[270](0.5,1.2){$\ldots$}
\uput[270](1.5,1.2){$\ldots$}
\uput[270](3.5,1.2){$\ldots$}
\psline[linewidth=\et]{-}(4,1)(5,2)
\psline[linewidth=\et]{-}(4,1)(5,0)

\uput[270](0,1){$1$}
\uput[270](1,1){$2$}
\uput[270](2,1){$j_1$}
\uput[270](3,1){$j_\ell$}
\uput[0](4,1){$n-2$}
\uput[0](5,2){$n-1$}
\uput[0](5,0){$n$}
\end{pspicture}
\caption{Positions of $j_1$ and $j_\ell$ arising in subcase (1.4.1).} \label{fig:tree2:in:proof}
\end{figure}
\end{center}
(1.4.2)  Suppose $j_1 > j_\ell$ and $j_1,j_\ell \le n-2$.  Then, $d(n,j_1)+d(j_\ell,1) \le d(1,n)$, and so the inequality holds.

(1.4.3) Suppose $j_1=n-1$.  Then $d(n,j_1)=2$.  Also, $d(j_\ell,1) \le d(n,1)$ and $d(j_\ell,j_1) \ge 1$, and so the inequality holds.

(1.4.4)  Suppose $j_\ell = n-1$.  Then, $d(j_\ell,1)=d(n,1)$ and $d(n,j_1) = d(j_\ell,j_1)$, and so again the inequality holds.

 (1.5) Suppose $s=1, \ell=1$, so $\sigma=(1,i,n,j) \hat{\sigma}$.  Let $\sigma'=(1,n)(i,j) \hat{\sigma}$.

 (1.5.1)  If $i=n-1$, by symmetry in $T$ between vertices $n$ and $n-1$, this subcase is resolved by subcase (1.4).

 (1.5.2) Let $2 \le i \le n-2$.  Then $d(1,i)+d(i,n)=d(1,n)$. So $f_T(\sigma) \le f_T(\sigma')$ iff $d(n,j)+d(j,1) \le d(1,n)+d(i,j)+d(j,i)+1$, which is true since $d(n,j)+d(j,1) \le d(1,n)+2$.

(1.6) Suppose $s=1, \ell \ge 2$. So $\sigma=(1,i,n,j_1,\ldots,j_\ell)\hat{\sigma}$.  Let $\sigma' = (1,n)(i,j_1,\ldots,j_\ell) \hat{\sigma}$.  It suffices to show that $f_T(\sigma) \le f_T(\sigma')$, i.e., that $d(1,i)+d(i,n)+d(n,j_1)+d(j_\ell,1) \le d(1,n)+d(1,n)+d(i,j_1)+d(j_\ell,i)+1$.  We examine the terms of this latter inequality for various subcases:

(1.6.1) Suppose $2 \le i \le n-2$.  Then $d(1,i)+d(i,n)=d(1,n)$.

(1.6.1a)  If $j_\ell = n-1$, then $d(j_\ell,i)=n-i-1$ and $d(i,j_1)=|i-j_1|$, and so the inequality holds iff $-1 \le |j_1-i|+j_1-i$, which is clearly true.

(1.6.1b) Suppose $2 \le j_\ell \le n-2$.  Then, the inequality holds iff $d(n,j_1)+j_\ell-1 \le n-2+|i-j_1|+|i-j_\ell|+1$, which can be verified separately for the cases $j_1=n-1$ and $2 \le j_1 \le n-2$.

(1.6.2) Suppose $i=n-1$.  By symmetry in $T$ of the vertices $n$ and $n-1$, this case is resolved by (1.4).

(1.7) Suppose $s \ge 2, \ell=0$, so $\sigma=(1,k_1,\ldots,k_s,n) \hat{\sigma}$.  Since $f_T(\sigma)=f_T(\sigma^{-1})$, this case is resolved by (1.4).

(1.8) Suppose $s \ge 2, \ell=1$, so $\sigma=(1,k_1,\ldots,k_s,n,j_1) \hat{\sigma}$.  Let $\sigma'=(1,n)(k_1,\ldots,k_s,j_1) \hat{\sigma}$.   We can assume that $2 \le k_1 \le n-2$ since the $k_1=n-1$ case is resolved by (1.4) due to the symmetry in $T$. To show $f_T(\sigma) \le f_T(\sigma')$, it suffices to prove the inequality $d(1,k_1)+d(k_s,n)+d(n,j_1)+d(j_1,1) \le d(1,n)+d(n,1)+d(k_s,j_1)+d(j_1,k_1)+1$.  We prove this inequality by separately considering whether $j_1=n-1$ or $k_s=n-1$ or neither:

(1.8.1) Suppose $j_1=n-1$.  Substituting $d(k_s,n)=n-k_s-1, d(n,j_1)=2, d(j_1,1)=j_1-1$, etc, we get that the inequality holds iff $k_1 \le |j_1-k_1|+n-2$, which is clearly true.

(1.8.2)  Suppose $2 \le j_1 \le n-2$.  Then $d(n,j_1)=n-j_1-1$, and so the inequality holds iff $k_1+d(k_s,n)+n-3 \le 2n-3+d(k_s,j_1)+d(j_1,k_1)$.  If $k_s=n-1$, this reduces to $j_1+k_1 \le 2n-3+|j_1-k_1|$, and is true, whereas if $2 \le k_s \le n-2$, this reduces to $k_1-k_s \le 1+|j_1-k_s|+|j_1-k_1|$, which is true due to the triangle inequality.

(1.9) Suppose $s,l \ge 2$, so $\sigma=(1,k_1,\ldots,k_s,n,j_1,\ldots,j_\ell) \hat{\sigma}$.

Let $\sigma'=(1,k_1,\ldots,k_s,n)(j_1,\ldots,j_\ell) \hat{\sigma}$.  It suffices to show that $S_T(\sigma) \le S_T(\sigma')+1$, i.e., that $d(n,j_1)+j_\ell \le n+d(j_1,j_\ell)$.

(1.9.1) If $j_1 < j_\ell$, then $j_1 \le n-2$, and so $d(n,j_1) =n-j_1-1$ and $d(j_1,j_\ell)=j_\ell-j_1$; the inequality thus holds.

(1.9.2) If $j_1 > j_\ell$, then $d(j_1,j_\ell)=j_1-j_\ell$, and so it suffices to show that $d(n,j_1) \le n+j_1-2j_\ell$.  It can be verified that this holds if $j_1=n-1$ and also if $2 \le j \le n-2$.

(2) Now suppose 1 and $n$ are in different cycles of $\sigma$.  So let $\sigma=(1,k_1,\ldots,k_s)(n,j_1,\ldots,j_\ell) \hat{\sigma}$.

(2.1) Suppose $s=0$.  Then $f_T(\sigma) \le {{n-1} \choose 2}-2$, by induction on $n$.

(2.2) Suppose $s=1$.  So let $\sigma=(1,i)(n,j_1,\ldots,j_\ell)\hat{\sigma}$.  By symmetry in $T$ between vertices $n$ and $n-1$ and subcase (1.1), we may assume $i \ne n-1$. Let $\sigma' = (1,n)(i,j_1,\ldots,j_\ell)\hat{\sigma}$.  It suffices to show that $S_T(\sigma) \le S_T(\sigma')$.  If $\ell=0$ this is clear since $d(1,i) \le d(1,n)$.  Suppose $\ell \ge 2$.  Then, by the triangle inequality, $d(n,j_1)+d(j_\ell,n) \le d(j_1,i)+d(i,n)+d(i,j_\ell)+d(i,n) = d(j_1,i)+d(i,j_\ell)+(n-i-1)2$.  Also, $d(1,n)=d(1,i)+d(i,n) = d(1,i)+n-i-1$.  Hence, $2d(1,i)+d(n,j_1)+d(j_\ell,n) \le 2d(1,n)+d(i,j_1)+d(j_\ell,i)$.  Hence, $S_T(\sigma) \le S_T(\sigma')$.  The case $\ell=1$ can be similarly resolved by substituting $j_1$ for $j_\ell$ in the $l\ge2$ case here.

(2.3) Suppose $s \ge 2, \ell=0$.  Then, by Theorem~\ref{Thm:pathfTmax:equals:diam}, $f_T(\sigma) \le {{n-1} \choose 2}$.

(2.4) Suppose $s \ge 2, \ell=1$, so $\sigma=(1,k_1,\ldots,k_s)(n,j_1) \hat{\sigma}$.  Let $\sigma'=(1,n)(k_1,\ldots,k_s,j_1)\hat{\sigma}$.  It suffices to show that $d(1,k_1)+d(1,k_s)+2d(n,j_1) \le 2d(1,n)+d(k_s,j_1)+d(k_1,j_1)$.  This inequality is established by considering the two subcases:

(2.4.1)  Suppose $j_1=n-1$.  Then the inequality holds iff $2k_s+k_1 \le 3n-7+|k_1-j_1|$, which is true since $k_1,k_2 \le n-2$ and $|k_1-j_1| \ge 1$.

(2.4.2) Suppose $j_1 \ne n-1$.  Then the inequality holds iff $k_1-j_1+k_s-j_1 \le |k_1-j_1|+|k_s-j_1|$, which is clearly true.

(2.5) Suppose $s, \ell \ge 2$, so $\sigma=(1,k_1,\ldots,k_s)(n,j_1,\ldots,j_\ell)\hat{\sigma}$.

Let $\sigma'=(1,n)(k_1,\ldots,k_s,j_1,\ldots,j_\ell)\hat{\sigma}$.  To show $f_T(\sigma) \le f_T(\sigma')$, it suffices to show that $d(1,k_1)+d(k_s,1)+d(n,j_1)+d(j_\ell,n) \le 2d(n,1)+d(k_s,j_1)+d(j_\ell,k_1)$.  By symmetry in $T$ between vertices $n$ and $n-1$, we may assume $k_1,\ldots,k_s \ne n-1$, since these cases were covered in (1).  We establish this inequality as follows:

(2.5.1)  Suppose $j_1=n-1$.  Then $d(n,j_1)=2$ and $d(n,j_\ell)=n-j_\ell-1$.  So the inequality holds iff $2k_s \le 2(n-2)+|j_\ell-k_1|+j_\ell-k_1$, which is true since $k_s \le n-2$ and $|j_\ell-k_1|+j_\ell-k_1 \ge 0$.

(2.5.2)  Suppose $j_1 \ne n-1$.  Then $d(n,j_1)=n-j_1-1$.  If $j_\ell=n-1$, the inequality holds iff $2 k_1 \le 2(n-2)+j_1-k_s+|j_1-k_s|$,which is true since $k_1 \le n-2$.  If $j_\ell \ne n-1$, the inequality holds iff $k_s-j_1+k_1-j_\ell \le |k_s-j_1|+|k_1-j_\ell|$, which is true.

This concludes the first part of the proof.

We now provide the second part of the proof.  Let $\Gamma$ be the Cayley graph generated by $T$.  We show that $\diam(\Gamma) \le {{n-1} \choose 2}+1.$  Let $\pi \in S_n$, and suppose each vertex $i$ of $T$ has marker $\pi(i)$.  We show that all markers can be homed using at most the proposed number of transpositions.  Since $\diam(T)=n-2$, marker 1 can be moved to vertex 1 using at most $n-2$ transpositions.  Now remove vertex 1 from the tree $T$, and repeat this procedure for marker 2, and then for marker 3, and so on, removing each vertex from $T$ after its marker is homed.  Continuing in this manner, we eventually arrive at a star $K_{1,3}$, whose Cayley graph has diameter 4.  Hence, the diameter of $\Gamma$ is at most $[(n-2)+(n-3)+\ldots+5+4+3]+4 = {{n-1} \choose 2}+1$.  This completes the proof.
\end{proof}

Let $s(n)$ denote the number of non-isomorphic trees on $n$ vertices and let $h(n)$ denote the number of nonisomorphic trees on $n$ vertices for which the diameter upper bound is sharp.  Let $\Delta_n$ be the strictness as defined above, and let $\gamma(n)$ denote the number of nonisomorphic trees on $n$ vertices for which the difference between the diameter upper bound and the true diameter is equal to $\Delta_n$.  Then, computer simulations yield the results in Table~\ref{table:summary}:
\begin{table}[ht]
\caption{Number of trees for which the bound $f(T)$ is sharp, and strictness} 
\centering 
\begin{tabular}{c | c c c c c} 
$n$ & 5 & 6 & 7 & 8 & 9 \\
\hline 
$s(n)$ & 3 & 6 & 11 & 23 & 47 \\
$h(n)$ & 2 & 4 & 3 & 6 & 4 \\
$\Delta_n $ & 1 & 2 & 3 & 4 & 6 \\
$\gamma_n $ & 1 & 1 & 1 & 3 & 2
\end{tabular}
\label{table:summary} 
\end{table}

These results imply that the $n-4$ lower bound for $\Delta_n$ is best possible, and an open problem is to obtain an upper bound for $\Delta_n$.  Another open problem is to classify the remaining families of trees for which the diameter upper bound is sharp.

\section{An algorithm for the diameter of Cayley graphs generated by transposition trees}
\label{sec:algorithm}

We now provide an algorithm that takes as its input a transposition tree $T$ on $n$ vertices and provides as output an estimate of the diameter of the Cayley graph (on $n!$ vertices) generated by $T$.  We then prove that the value obtained by our algorithm is less than or equal to the previously known diameter upper bound $f(T)$.  The notation used to describe our algorithm should be self-explanatory and is similar to that used in Knuth \cite{Knuth:2011}.

\bigskip
\noindent \textbf{Algorithm A}
\\Given a transposition tree $T$, this algorithm computes a value $\beta$ which is an estimate for the diameter of the Cayley graph generated by $T$.  $|V(T)|$ denotes the current value of the number of vertices in $T$; initially, $V(T)=\{1,2,\ldots,n\}$.
\\ \textbf{A1.} [Initialize.]  \\Set $\beta \leftarrow 0$.
\\ \textbf{A2.} [Find two vertices $i,j$ of $T$ that are a maximum distance apart.]  \\Find any two vertices $i,j$ of $T$ such that $\dist_T(i,j)=\diam(T)$.
\\ \textbf{A3.} [Update $\beta$, and remove $i,j$ from $T$.] \\Set $\beta \leftarrow \beta+(2 \diam(T)-1)$, and set $T \leftarrow T-\{i,j\}$.  If $T$ still has 3 or more vertices, return to step A2; otherwise, set $\beta \leftarrow \beta+|V(T)|-1$ and terminate this algorithm. \qed

\bigskip One way to implement step A2, which picks any two vertices of the tree that are a maximum distance apart, is as follows.  Assign label 1 to each leaf vertex, and then remove all the leaf vertices from the tree.  Then all the leaf vertices of the smaller tree can be assigned label 2, and so on, until we arrive at the center of the tree, which has (exactly one or two) vertices of the highest label.  One can then start at the center and use the stored labels to construct a path of maximum length in the tree.

We now show that the value obtained by Algorithm A is less than or equal to the previously known diameter upper bound.

\begin{Theorem} \label{beta:le:AKbound}
Let $T$ be a transposition tree on vertex set $\{1,2,\ldots,n\}$, and let $\beta$ be the value obtained by Algorithm A for input $T$. Then,
$$\beta \le \max_{\pi \in S_n}  \left\{c(\pi)-n+\sum_{i=1}^n \dist_T(i,\pi(i)) \right\}.$$
\end{Theorem}

\begin{proof}
Let $\{i_1,j_1\},\{i_2,j_2\},\ldots,\{i_r,j_r\}$ be the vertex pairs chosen by Algorithm A during the $r$ iterations of step A2.  We now construct a permutation $\pi$ as follows.  If $n$ is odd, then $T$ contains only one vertex, say $i_{r+1}$ when the algorithm terminates.  In this case, we let $\pi = (i_1,j_1) \ldots (i_r,j_r)(i_{r+1}) \in S_n$.  If $n$ is even, then $T$ contains two vertices, say $i_{r+1},j_{r+1}$, when the algorithm terminates.  In this case, we let $\pi=(i_1,j_1) \ldots (i_{r+1},j_{r+1}) \in S_n$.  In either case, $r+1 = \lceil n/2 \rceil$, the value of $\beta$ computed by the algorithm equals
$$\beta=\left(\sum_{\ell=1}^r \left\{2 \dist_T(i_\ell,j_\ell)-1 \right\} \right)+\left[ (n+1) \mod 2 \right],$$
and the quantity $f_T(\pi) := c(\pi)-n+\sum_{i=1}^n \dist_T(i,\pi(i))$ evaluates to
$$f_T(\pi)=(r+1)-n+ \left(2 \sum_{\ell=1}^r \dist_T(i_\ell,j_\ell)+2 \left[(n+1) \mod 2 \right] \right).$$
A quick check shows that the two expressions for $\beta$ and $f_T(\pi)$ are equal.  Hence, $\beta \le \max_{\pi \in S_n} f_T(\pi)$.
\end{proof}

Note that we have not established that the value computed by Algorithm A is unique (and in fact, it isn't sometimes); for a given tree, there can exist more than one pair of vertices that are a maximum distance apart, and different vertex pairs chosen during step A2 can sometimes yield different values of $\beta$.  Let $\mathcal{B}$ denote the set of possible values that can be the output of Algorithm A, and let $\beta_{\max} := \max_{\beta \in \mathcal{B}} \beta$. It follows immediately from Theorem~\ref{beta:le:AKbound} that $\beta_{\max}$ is less than or equal to the previously known diameter upper bound.  On the other side, we now show that $\beta_{\max}$ is lower bounded by the true diameter of the Cayley graph $\Gamma$:

\begin{Theorem}
Let $\Gamma$ be the Cayley graph generated by a transposition tree $T$.  Let $\beta_{\max}$ be as defined above and equal to the maximum possible value returned by Algorithm A.  Then,
$$\diam(\Gamma) ~\le~ \beta_{\max}~ \le~ \max_{\pi \in S_n} \left\{c(\pi)-n+\sum_{i=1}^n \dist_T(i,\pi(i)) \right\}.$$
\end{Theorem}

\begin{proof}
The second inequality has already been proved.  We now prove the first inequality.  Let $\pi \in S_n$.  Suppose that each vertex $k$ of $T$ initially has marker $\pi(k)$. It suffices to show that all markers can be homed to their vertices using at most $\beta_{\max}$ transpositions.

Consider the following procedure.  Pick any two vertices $i,j$ of $T$ that are a maximum distance apart. We consider two cases, depending on the distance in $T$ between vertex $i$ and the current location $\pi^{-1}(i)$ of the marker $i$:

 Suppose that the distance in $T$ between vertices $i$ and $\pi^{-1}(i)$ is at most $\diam(T)-1$.  Then marker $i$ can be homed using at most $\diam(T)-1$ transpositions.  And then, marker $j$ can be homed using at most $\diam(T)$ transpositions.  Hence, markers $i$ and $j$ can both be homed to leaf vertices $i$ and $j$, respectively, using at most $2 \diam(T)-1$ transpositions.  We now let $i_1=i$ and $j_1=j$.

Now consider the case where the distance in $T$ between vertices $i$ and $\pi^{-1}(i)$ is equal to $\diam(T)$.  Let $x$ be the unique vertex adjacent to $\pi^{-1}(i)$.  In the first sequence of steps, marker $\pi^{-1}(i)$ can be homed to vertex $\pi^{-1}(i)$ using at most $\diam(T)$ transpositions.  The last of these transpositions will home marker $\pi^{-1}(i)$ and place marker $i$ at vertex $x$, whose distance to $i$ is exactly $\diam(T)-1$.  In the second sequence of steps, marker $i$ can be homed to vertex $i$ using at most $\diam(T)-1$ transpositions.  Hence, using these two sequences of steps, markers $i$ and $\pi^{-1}(i)$ can both be homed using at most $2 \diam(T)-1$ transpositions.  We now let $i_1=i$ and $j_1=\pi^{-1}(i)$.

We now remove from $T$ the vertices $i_1$ and $j_1$, and repeat this procedure on $T-\{i_1,j_1\}$ to get another pair $\{i_2,j_2\}$.  Continuing in this manner until $T$ contains at most two vertices, we see that all markers can be homed using at most
$$ \left\{ \sum_{\ell=1}^r \left( 2 \dist_T(i_\ell,j_\ell)-1 \right)\right\} + \left[(n+1) \mod 2 \right]$$
transpositions.  This quantity is equal to the value $\beta$ returned by the Algorithm when it chooses $\{i_1,j_1\},\ldots,\{i_r,j_r\}$ as its vertex pairs during each iteration of step A2, and hence this quantity is at most $\beta_{\max}$.  Thus, $\dist_{\Gamma}(I,\pi) \le \beta_{\max}$ for all $\pi \in S_n$.
\end{proof}

We have shown that the maximum possible value returned by the algorithm, denoted by $\beta_{\max}$, is an upper bound on the diameter of the Cayley graph. An open problem is to determine whether {\em each} of the possible values returned by the algorithm is an upper bound on the diameter, i.e. whether $\beta$ is an upper bound on the diameter of the Cayley graph for each $\beta \in \mathcal{B}$. Our examples so far show that for many families of trees (in fact, for all the ones investigated so far), the value returned by the algorithm is an upper bound on the diameter. In other words, while we only showed that $\beta_{\max}$ is an upper bound on the diameter, it is certainly possible that in almost all cases any $\beta$ is also an upper bound on the diameter.

Note that our results imply that the value returned by Algorithm A is an upper bound on the diameter for all trees for which $|\mathcal{B}|=1$ since for such trees $\beta=\beta_{\max}$.  For such trees, our algorithm efficiently computes a value which is both an upper bound on the diameter as well as better than (or at least as good as) the previously known diameter upper bound.  An open question is to determine whether this is also the case for the remaining trees.  Trees for which $|\mathcal{B}| \ge 2$ are rare, and an open question is to determine whether almost all trees have $|\mathcal{B}|=1$.

We now discuss some properties of the algorithm.

Consider the transposition tree $T_1=\{(1,2),(2,3),(3,7),(7,8),(3,4),(4,5),(4,6)\}$.  If Algorithm A picks the sequence of vertex pairs during step A2 to be $\{1,8\},\{5,7\}$ and $\{2,6\}$, then the value returned by the algorithm is $\beta = 7+5+5+1=18$.  If Algorithm A picks the pairs to be $\{1,5\},\{6,8\}$ and $\{2.7\}$, then the value returned by the algorithm is still $\beta=7+7+3+1=18$.  In this example, the value returned by the algorithm is unique even though the subtrees $T-\{1,8\}$ and $T-\{1,5\}$ are non-isomorphic and even have different diameters.

Now consider the transposition tree $T_2=\{(1,2),(2,3),$ $(3,6),(3,4),(4,5),$ $(6,7),$ $(6,8),(6,9)\}$. If Algorithm A picks the sequence of vertex pairs during step A2 to be $\{1,5\},\{2,7\},\{4,8\}$ and $\{3,9\}$, then the value returned by the algorithm is $\beta=7+5+5+3=20$.  And if Algorithm A picks the vertex pairs to be $\{1,7\},\{5,8\},\{2,9\}$ and $\{4,6\}$, then the value returned by the algorithm is $\beta=7+7+5+3=22$.  Hence, the set of values returned by the algorithm $\mathcal{B} \supseteq \{20,22\}$.  A computer simulation shows the true diameter value of the Cayley graph generated by $T_2$ to be 18, and the diameter upper bound $f(T)$ evaluates to 22.

\begin{center}
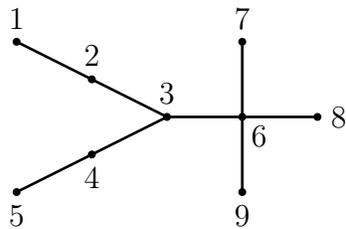
\begin{figure}
\begin{pspicture}(-3,0)(5,2)
\qdisk(0,2){\vs}
\qdisk(0,0){\vs}
\qdisk(1,1.5){\vs}
\qdisk(1,0.5){\vs}
\qdisk(2,1){\vs}
\qdisk(3,1){\vs}
\qdisk(3,0){\vs}
\qdisk(3,2){\vs}
\qdisk(4,1){\vs}
\psline[linewidth=\et]{-}(0,2)(1,1.5)
\psline[linewidth=\et]{-}(0,0)(1,0.5)
\psline[linewidth=\et]{-}(1,1.5)(2,1)
\psline[linewidth=\et]{-}(1,0.5)(2,1)
\psline[linewidth=\et]{-}(2,1)(3,1)
\psline[linewidth=\et]{-}(3,1)(3,2)
\psline[linewidth=\et]{-}(3,1)(3,0)
\psline[linewidth=\et]{-}(3,1)(4,1)
\uput[270](0,0){$5$}
\uput[90](0,2){$1$}
\uput[90](1,1.5){$2$}
\uput[270](1,0.5){$4$}
\uput[90](2,1){$3$}
\uput[315](3,1){$6$}
\uput[90](3,2){$7$}
\uput[0](4,1){$8$}
\uput[270](3,0){$9$}
\end{pspicture}
\caption{The transposition tree $T_2$ for which $|\mathcal{B}| \ge 2$.} \label{fig:tree:in:proof}
\end{figure}
\end{center}
In terms of the strictness of the diameter estimate computed by this algorithm, it is can be shown that for every $n$, there exists a tree on $n$ vertices, such that the difference between the value computed by this algorithm and the actual diameter value is, again, at least $n-4$.  The proof uses the same transposition tree constructed in the proof of the similar result for the strictness of the diameter upper bound, but unlike in the earlier proof, here it is short and immediate to establish that the (unique) value computed by the algorithm is ${n \choose 2}-2$.

These results raise some further questions.  Is it true that for all trees, the maximum possible value $\beta_{\max}$ returned by the algorithm is equal to the diameter upper bound?  Another open problem is to characterize those trees for which the value returned by the algorithm is unique, i.e. for which $|\mathcal{B}|=1$.  Another open problem is to classify those trees for which the sequences of subtrees generated by the algorithm are isomorphic (i.e. are independent of the choice of vertex pairs during step A2, unlike the tree $T_1$, as mentioned above).

\section{Concluding remarks}
A problem of much interest is to determine exact or approximate values for the diameter of various families of Cayley networks of permutation groups.  When a bound or algorithm for this diameter problem is proposed in the literature, it is of interest to know for which cases the bound is exact or how far away the bound can be from the actual value in the worst case. The diameter upper bound from \cite{Akers:Krishnamurthy:1989} was studied in this work.  This formula bounds the diameter of the Cayley graph on $n!$ vertices in terms of parameters of the underlying transposition tree. We showed above that this bound is sharp for all trees of minimum diameter and for all trees of maximum diameter, but can be strict for trees that are not extremal.  We also showed that for every $n >4$, there exists a tree on $n$ vertices such that the difference between the upper bound for the diameter of the Cayley graph on $n!$ vertices and its true diameter is at least $n-4$.

An open problem is to characterize (all) the remaining families of trees for which the diameter upper bound is sharp.  The $n-4$ lower bound for $\Delta_n$ given here is best possible and an open problem is to determine an upper bound for $\Delta_n$.

Evaluating the previously known upper bound on the diameter requires on the order of $n!$ (times a polynomial) computations.  We proposed an algorithm for the same problem which uses on the order of $n$ (times a polynomial) computations.  This was possible because we worked directly with the transposition tree on $n$ vertices, and so our algorithm does not require examining any of the permutations (it is only the proof of the algorithm that required examining the permutations).  As far as the accuracy of the diameter estimate of our algorithm, we showed that the value obtained by our algorithm is less than or equal to the previously known diameter upper bound.  For the families of trees investigated so far, the maximum possible value returned by our algorithm is exactly equal to the previously known diameter upper bound. However, our algorithm arrived at the same value using a very different (and also simpler and more efficient) method than the previously known diameter upper bound, and investigating the properties of this algorithm might lead to new insights on this problem.  There are many open problems and extensions on this algorithm and related bounds, some of which were discussed at the end of Section~\ref{sec:algorithm}.

\bibliographystyle{plain}
\bibliography{refscayley}

\end{document}